\documentclass[letterpaper,11pt]{article}
\usepackage{amsmath}
\usepackage{amsfonts}
\usepackage{xspace}
\usepackage{fullpage}
\usepackage{booktabs}
\usepackage{enumitem}
\usepackage{amsthm}
\usepackage{url}
\usepackage{xfrac}
\usepackage{csquotes}

\newcommand{\nc}{{NeuraCrypt}\xspace} 
\newcommand{\isim}{{\text{i-sim}}}
\newcommand{\psim}{{\text{p-sim}}}

\newcommand{\AUX}{\mathsf{AUX}}
\newcommand{\aux}{\mathsf{aux}}

\newcommand{\sanjam}[1]{}
\newcommand{\sedits}[1]{}
\newcommand{\FT}[1]{}
\newcommand{\SJ}[1]{}
\newcommand{\Extract}{\mathsf{Pred}}

\newcommand{\MohNote}[1]{}

\newcommand{\Snote}[1]{}

\newcommand{\AbhrNote}[1]{}
\newcommand{\SamNote}[1]{}
\newcommand{\SSNote}[1]{}

\newcommand{\fhat}[2]{\ifthenelse{\equal{#2}{}}{\hat{f}(#1)}{\ifthenelse{\equal{#2}{0}}{\hat{f}(\emptyset)}{\hat{f}(#1_{\leq #2})}}}
\newcommand{\ftild}[2]{\ifthenelse{\equal{#2}{}}{\tilde{f}(#1)}{\ifthenelse{\equal{#2}{0}}{\tilde{f}(\emptyset)}{\tilde{f}(#1_{\leq #2})}}}
\newcommand{\ftildstar}[2]{\ifthenelse{\equal{#2}{}}{\tilde{f^*}(#1)}{\ifthenelse{\equal{#2}{0}}{\tilde{f^*}(\emptyset)}{\tilde{f^*}(#1_{\leq #2})}}}
\newcommand{\ghat}[2]{\ifthenelse{\equal{#2}{}}{\hat{g}(#1)}{\ifthenelse{\equal{#2}{0}}{\hat{g}(\emptyset)}{\hat{g}(#1_{\leq #2})}}}
\newcommand{\mfix}[2]{\ifthenelse{\equal{#2}{}}{m(#1)}{\ifthenelse{\equal{#2}{0}}{m(\emptyset)}{m(#1_{\leq #2})}}}

\newcommand{\fapp}[2]{\ifthenelse{\equal{#2}{}}{\tilde{f}(#1)}{\ifthenelse{\equal{#2}{0}}{\tilde{f}(\emptyset)}{\tilde{f}(#1_{\leq #2})}}}

\newcommand{\Risk}{\mathsf{Risk}}

\newcommand{\aSF}{\mathsf{a}}
\newcommand{\gSF}{\mathsf{g}}
\newcommand{\hSF}{\mathsf{h}}

\newcommand{\avr}[2]{\ifthenelse{\equal{#2}{}}{\aSF({#1})}{\ifthenelse{\equal{#2}{0}}{\aSF(\emptyset)}{\aSF({#1}_{\leq #2})}}}

\newcommand{\avrMax}[2]{\ifthenelse{\equal{#2}{}}{\aSF^*({#1})}{\ifthenelse{\equal{#2}{0}}{\aSF^*(\emptyset)}{\aSF^*({#1}_{\leq #2})}}}

\newcommand{\avrApp}[2]{\ifthenelse{\equal{#2}{}}{\tilde{\aSF}({#1})}{\ifthenelse{\equal{#2}{0}}{\tilde{\aSF}(\emptyset)}{\tilde{\aSF}({#1}_{\leq #2})}}}

\newcommand{\avrAppMax}[2]{\ifthenelse{\equal{#2}{}}{\tilde{\aSF}^*({#1})}{\ifthenelse{\equal{#2}{0}}{\tilde{\aSF}^*(\emptyset)}{\tilde{\aSF}^*({#1}_{\leq #2})}}}

\newcommand{\ArgMax}[2]{\ifthenelse{\equal{#2}{}}{\hSF({#1})}{\ifthenelse{\equal{#2}{0}}{\hSF(\emptyset)}{\hSF({#1}_{\leq #2})}}}

\newcommand{\AppArgMax}[2]{\ifthenelse{\equal{#2}{}}{\tilde{\hSF}({#1})}{\ifthenelse{\equal{#2}{0}}{\tilde{\hSF}(\emptyset)}{\tilde{\hSF}({#1}_{\leq #2})}}}

\newcommand{\gain}[2]{\ifthenelse{\equal{#2}{}}{\gSF(#1)}{\gSF(#1_{\leq #2})}}
\newcommand{\gainMax}[2]{\ifthenelse{\equal{#2}{}}{\gSF^*(#1)}{\gSF^*(#1_{\leq #2})}}
\newcommand{\gainApp}[2]{\ifthenelse{\equal{#2}{}}{\tilde{\gSF}(#1)}{\tilde{\gSF}(#1_{\leq #2})}}
\newcommand{\gainAppMax}[2]{\ifthenelse{\equal{#2}{}}{\tilde{\gSF}^*(#1)}{\tilde{\gSF}^*(#1_{\leq #2})}}

\newcommand{\pr}[2][]{\Pr_{\ifthenelse{\isempty{#1}}{}{{#1}}}\left[{#2}\right]}

\newcommand{\remove}[1]{}

\newcommand{\set}[1]{\left\{ #1 \right\}}

\newcommand{\R}{{\mathbb R}}
\newcommand{\N}{{\mathbb N}}

\newcommand{\cD}{{\mathcal D}}

\newcommand{\cK}{{\mathcal K}}

\newcommand{\eps}{\varepsilon}

\newcommand{\poly}{\operatorname{poly}}

\makeatletter
\def\th@protocol{%
    \normalfont 
    \setbeamercolor{block title example}{bg=orange,fg=white}
    \setbeamercolor{block body example}{bg=orange!20,fg=black}
    \def\inserttheoremblockenv{exampleblock}
  }
\makeatother

\newcommand{\namedref}[2]{#1~\ref{#2}}

\newcommand{\torestate}[3]{%
\expandafter \def \csname BBRESTATE #2 \endcsname{#3}
\theoremstyle{plain}
\newtheorem{BBRESTATETHMNUM#2}[theorem]{#1}
\begin{BBRESTATETHMNUM#2}\label{#2}\csname BBRESTATE #2 \endcsname   \end{BBRESTATETHMNUM#2}
\newtheorem*{BBRESTATETHMNONNUM#2}{\namedref{#1}{#2}}
}

\newcommand{\restate}[1]{\begin{BBRESTATETHMNONNUM#1}[Restated] \csname BBRESTATE #1 \endcsname
\end{BBRESTATETHMNONNUM#1}}

\usepackage{xcolor}

\newtheorem{definition}{Definition}
\newtheorem{theorem}{Theorem}
\newtheorem{remark}{Remark}

\usepackage{authblk}

\begin{document}

\title{\nc is not private}

\author[1]{Nicholas Carlini$^*$}
\author[2,3]{Sanjam Garg}
\author[4]{Somesh Jha}
\author[5]{Saeed Mahloujifar}
\author[6]{\\ Mohammad Mahmoody}
\author[7,1]{Florian Tramer}

\affil{Google, $^2$UC Berkeley, $^3$NTT Research, $^4$University of Wisconsin, $^5$Princeton University, $^6$University of Virginia, $^7$Stanford University}

\date{}
\maketitle

\newcommand\blfootnote[1]{%
  \begingroup
  \renewcommand\thefootnote{}\footnote{#1}%
  \addtocounter{footnote}{-1}%
  \endgroup
}
\blfootnote{$^*$ Authors ordered alphabetically.}

\section*{Abstract}
\nc (Yara et al. arXiv 2021) is an algorithm that converts a sensitive dataset to
an encoded dataset so that
(1) it is still possible to train machine learning models on
the encoded data, but
(2) an adversary who has access only to the encoded dataset can
not learn much about the original sensitive dataset.
We break \nc's privacy claims, by perfectly solving the authors' public challenge, and by showing that \nc does not satisfy
the formal privacy definitions posed in the original paper.
Our attack consists of a series of boosting steps that,
coupled with various design flaws,
turns a $1\%$ attack advantage into a $100\%$ complete break of the scheme.

\section{Introduction}
In order to train neural networks on sensitive datasets (such as
medical images~\cite{hosny2018artificial,wernick2010machine,esteva2017dermatologist} or personal messages~\cite{chen2019gmail}) it is necessary
that the models be \emph{privacy-preserving}.
Given access to the trained model, it should not be possible to
learn anything about the training dataset.
One approach to training models that preserve privacy
``encodes'' each input with an encoding function
$e : \mathcal{X} \to \mathcal{Y}$ that
maps an original dataset $\mathcal{X}$ to an encoded dataset $\mathcal{Y}$ \cite{huang2020instahide}. 
The encoding should satisfy two properties:
\begin{enumerate}
\item \textbf{Utility:} A learning algorithm can use $\mathcal{Y}$ to train a useful
  model that is (approximately) as good as if the original dataset
  $\mathcal{X}$ was used instead.
\item \textbf{Privacy:} It is not possible to study the encoded dataset $\mathcal{Y}$ to
  learn nontrivial and sensitive properties about the original dataset $\mathcal{X}$.
\end{enumerate}

Encoding schemes are exciting because they allow 
\emph{any} training algorithm to run on the encoded dataset,
making special-purpose privacy-preserving training techniques unnecessary.

\subsection{\nc}
\nc~\cite{yala2021neuracrypt} is an encoding technique that
aims to achieve utility with privacy.
%
\nc encodes images in the training set
one at a time by running them forward through
a neural network with random weights.
The encoded outputs are directly the output of this model
after a pixel-block permutation.

Let $x \in \mathcal{X}$ be a $w \times h \times c$ dimensional image.
\nc first splits the image into $a^2$ patches of size ${w \over a} \times {h \over a} \times c$.
Each patch is then processed independently through a series of linear
and nonlinear transformations.
Specifically, each patch is first flattened into a vector
$\hat{x} \in \mathbb{R}^{{w \over a} \cdot {h \over a} \cdot c}$.
Then, \nc transforms each vector $\hat z_i = (f_{k} \circ ReLU \circ f_{k-1} \circ \dots \circ ReLU \circ f_{1})(\hat{x})$.
Each transform $f_i(x) = A_i x + b_i$ is a linear projection with randomly
initialized weights, sampled from a Normal distribution,
and $ReLU(x) = \max(x, 0)$.
We let $g : \hat x \to \hat z$ denote this patch encoding.

NeuraCrypt then stops processing at the patch-level, and begins
processing at the image-level.
Given the ordered set of patches $\{\hat z_i\}_{i=1}^{a^2}$,
\nc then adds a ``positional encoding'' $\delta_i$ that is different for
each patch in the image, giving a new set of images
$\tilde z_i = \hat z_i + \delta_i$ 
and performs a final linear projection
to obtain the encoded images $\hat{y}_i = f_{k+1}(ReLU(\tilde z_i))$.
\nc finally randomly permutes the individual patches $\hat{y}_i$ with a fresh random permutation $\pi$,
and returns $y_{\pi(i)}$.

\begin{table}
\centering
\begin{tabular}{@{}llll@{}}
\toprule
  \textbf{Security} & & \textbf{Number} & \textbf{Percent} \\
  \textbf{Parameter} & \textbf{Dataset} & \textbf{Of Images} & \textbf{Correct} \\
  \midrule
   $k=2$ & Ours (ImageNet) & 10,000 & 100\% \\
   $k=7$ & Ours (ImageNet) & 10,000 & 100\% \\
   $k=15$ & Ours (ImageNet) & 10,000 & 100\% \\
   \midrule
   $k=2$ & \nc Challenge (CheXpert) & 14,643 & 100\% \\
   $k=7$ & \nc Challenge (CheXpert) & 14,643 & 100\% \\
   \bottomrule
\end{tabular}
\caption{We completely break NeuraCrypt's privacy claims by constructing an algorithm that
can match original images to encoded images perfectly, with $100\%$ probability.
Our attacks work on our own dataset (ImageNet \cite{deng2009imagenet}) and the NeuraCrypt challenge 
(CheXpert medical images \cite{irvin2019chexpert})}
\label{tab:main}
\end{table}

This process is repeated for each image in the dataset, using the same random neural
network and the same positional encoding.
The \nc privacy claim is that these encodings preserve the privacy of their inputs. We show this is not the case.
We do not study the utility of \nc---it is not
obvious that processing patches of an image produces accurate models, but the
authors find it does.

\subsection{Privacy Game}

We evaluate privacy under the \nc Challenge \cite{nc_challenge}.

\smallskip \noindent 
\textbf{Setup.}
Let $\mathcal{X}$ be a dataset of $N$ unlabeled images
and $\text{sym}(\mathcal{X})$ an exponentially large family of encoding functions .
The two participants,
Alice and Bob, are both given $\mathcal{X}$ and $\text{sym}(\mathcal{X})$.

\smallskip \noindent 
\textbf{Alice} chooses a random encoding transform
$T \in \text{sym}(\mathcal{X})$ and chooses a random ordered subset
$\vec x = \{x_i\} \subset \mathcal{X}$.
Alice encodes $\{y_i\}$ by computing $y_i = T(x_i)$.
Alice chooses a random matching $\sigma$ and sends
to Bob the ordered set $\vec y = \{y_{\sigma_i}\}_i$.

\smallskip \noindent 
\textbf{Bob} receives $\vec x$ and $\vec y$ from Alice, 
and runs some attack to generate his guess of the matching
$\tilde\sigma = \mathcal{A}(\vec x,\vec y)$.
Bob's ``score'' is equal to the number of
entries where $\tilde \sigma$ correctly matches $\sigma$.
For a secure scheme, Bob should expect to score just $1$, and for $\lVert \mathcal{X} \rVert$
moderately large then $\text{Pr}[\text{Bob's score} \ge s] \approx {1 \over {s!}}$.

\subsection{Results}

We solve the above privacy game for \nc.
Our attack recovers the  entire mapping nearly perfectly, and as part
of the attack, also recovers Alice's transformation $T$---allowing further attacks if future images are encoded.
Prior instance-encoding schemes~\cite{huang2020instahide} were also shown to be not private through complete reconstruction attacks~\cite{carlini2020attack}.

Table~\ref{tab:main} gives our main results;
regardless of the size of the security parameter, or of the number of
images that have been encoded, we achieve a $100\%$ attack success rate.
We developed our attack exclusively using the ImageNet dataset,
and then evaluated on the NeuraCrypt Challenge once it was released.

\section{Our Attack}

We break \nc with a series of steps that iteratively boost
the adversary's advantage from $0$ all the way to a perfect $N$.
Our attack is dominated by a quadratic-time all-pairs comparison
between original and encoded images. We solve the \nc challenge
in under $6$ hours wall-clock time on a single machine. The attack steps are as follows:
%

\begin{enumerate}[leftmargin=3em]
\item Weakly break \nc's privacy.
\begin{enumerate}
\item[\S\ref{s1}] Construct a \emph{patch} similarity function $\psim$ that 
  returns $1$ if a patch $\hat y$ was generated by
  an original patch $\hat x$, and $0$ otherwise.

\item[\S\ref{s2}] Construct an \emph{image} similarity function $\isim$ that
  detects if an entire image $x$ was used to generate an entire
  image $y$.
  The resulting image similarity function has $>95\%$ accuracy.

\item[\S\ref{s3}] Match each $x_i \to y_{m(i)}$. We construct matching $m$
  by solving the minimum cost bipartite matching from each original
  image $x$ to each encoding $y$, with edge costs
   $\isim(x,y)$.
  This matching is a partial break of \nc and matches in more than
  a few percent of positions.

\end{enumerate}

\item Strongly break \nc's privacy.
\begin{enumerate}
\item[\S\ref{s4}] Recover the exact permutation $\pi_j$ used to encode each image
  $y_j$ by leveraging the fact that the ``position encoding'' leaks
  information.
  This step only recovers $\tilde \pi_j$ up to a single global permutation $\pi_g$ (i.e., so $\pi_j = \sigma \circ \tilde \pi_j$).

\item[\S\ref{s5}] Recover the global permutation
  $\sigma$ using image-specific local permutations $p^{(j)}$,
  and the approximate matching $m$
  by solving another matching problem on the
  image-encoding matching.
  %

\item[\S\ref{s6}] Improve the matching by repeating the prior two steps.
with correctly
  permuted images.
  The matching here is often correct in almost all positions.
\end{enumerate}

\item Completely break \nc's privacy.
\begin{enumerate}
\item[\S\ref{s7}] Extract the transformation $T$ by gradient descent,
  training on the now-approximately-aligned images and encodings.

\item[\S\ref{s8}] Recover the exact global matching by constructing a matching
  problem with weights given by similarity between encoded images
  $y_j$ and our reconstruction of them $\tilde T(x_i)$.
\end{enumerate}

\end{enumerate}
  
We implement our attack in $700$ lines of JAX, and make
our open source code available.

\subsection{Recover initial matching}

\noindent
\subsubsection{Compute patch similarity}
\label{s1}

\nc first maps an image $x$ into a set of patches $\{\hat x_i\}_{i=1}^{a^2}$,
and then encodes each patch with the transformation $\hat y_i = T(\hat x_i)$
(where $T$ is the defender's randomly chosen transform).
Our first step in breaking \nc learns a similarity function
$\psim(x, y)$ that outputs the probability that we believe
$\hat y = g(\hat x)$, as opposed to some other $\hat y = g(\hat x')$.

To do this we first construct a large dataset $D = \{(\hat x, t(\hat x))\}$
of patches with their corresponding encoded output, by sampling random encoders $t$.

We then construct two neural networks $n_x : \hat x \to \mathbb{R}^{d}$
and $n_y : \hat y \to \mathbb{R}^d$ that embed the original image
and encoded image into the same $d$-dimensional embedding space.
We train these neural networks so that
$\ell(\hat x, \hat y) = \langle n_x(\hat x), n_y(\hat y) \rangle$ is small when $(\hat x,\hat y) \in \mathcal{D}$.
Unfortunately training exclusively on this loss would give a degenerate
solution $n_x \equiv n_y \equiv 0$.
We address this by additionally requiring that $\ell(\hat x, \hat y')$ is
large when $\hat x$ and $\hat y'$ are unrelated patches.
Specifically, we train with a k-way contrastive softmax loss.

A minor, but important, detail is that instead of the neural networks
operating on the image pixels directly, we first compute the $k$th order
moments
with the $0$th order moment being the
mean $\mu_0 = {1 \over n} \sum_{i} v_i$, and the remaining higher order
moments being defined as
$\mu_k = {1 \over n}\sum_{i=1}^n (v_i - \mu_0)^k$.
This allows us to compute the moment vector over both image patches
(by flattening the vector) and over encoded vectors (which are already flat).
Our patch similarity function reaches up to $80\%$ accuracy
on a balanced dataset of patches that match and don't (i.e., $50\%$ is random guessing).

\subsubsection{Compute image similarity}
\label{s2}

Given the patch similarity function $\psim$ we then construct
an image similarity function $\isim$.
Suppose we knew the permutation $\pi$ that \nc applied to the encoded image patches.
Then we can extend our patch similarity function to an entire image, by simply averaging the similarity across all patches:
\[\isim_{\pi}(x, y) = {1 \over a^2} \sum_{i =1}^{a^2} \psim(x^{(i)}, y^{(\pi(i))}).\]

Even though we do not know which permutation $\pi$ was used,
we can still apply this idea.
If image $x$ actually was used to construct $y$, then
while we can not know where the patch $\hat x^{(0)}$ maps onto,
we do know that it maps onto \emph{some} $\hat y^{(j)}$.

Therefore, we can construct a bipartite graph with one side
containing original image patches $\hat x$ and the other side containing encoded patches $\hat y$.
The weight along the edge between a pair of patches is given
by $\psim(\hat x,\hat y)$.
%
The minimum cost maximum weight matching in this bipartite graph
is then the best way to pair together the embeddings, and we define
the cost of this matching to be the output of our $\isim$ function.
Viewed differently, this is just an efficient algorithm to compute:
$\max_{\rho \in Sym(a^2)} \isim_{\rho}$.

\subsubsection{Recover image-encoding matches}
\label{s3}
Given the matching function $\isim$, it is now trivial to match
each image $x_i$ with an encoded image $y_j$ by computing
$\isim(x_i, y_j)$ for all $N^2$ pairs of images, constructing
a cost matrix, and again solving the minimum cost bipartite matching.
Denote this matching by $m$ so that $m : i \to j$.

Note that this matching is not going to be very high quality,
because the image similarity matching is imperfect.%
\footnote{Even with $99\%$ accuracy, there are many false positives due to the low base rate: only $1$ of the $10,000$ original images is a true positive.}
However, as long as the accuracy is better than random chance, it
will suffice for our purposes.

\subsection{Boost to high quality matching}

\subsubsection{Recover per-image permutations}
\label{s4}
We now 
recover the permutation $\pi_j$ that was used to shuffle the pixel
blocks in the encoding $y_j$, independent of any of the above steps.

It turns out that given two encoded patches $\hat y$ and $\hat y'$, the value
$\langle \hat y, \hat y' \rangle$ correlates with whether or not the patches
were placed in the same location with respect to the original image---whatever
(unknown) position that happened to be.
Therefore, to recover the per-image permutation, we choose
one encoded image as the \emph{reference} permutation and
decorrelate all other permutations with respect to this reference
permutation by computing each matching.

To see why this correlation occurs,
recall that \nc works by first converting the original image patches 
$\hat x_i$ into the
latent encoding $\hat z_i$.
Then, the pre-permuted output is given
by $\hat y_j = A_{k+1}(ReLU(z + \delta_j)) + b_{k+1}$.
where
$\delta_j$ is a (fixed at initialization, but randomly sampled)
vector $\delta_j \sim N(0,I \cdot 1)$.
If we were to pretend that $ReLU(z)=z$ (and it is, half of the time),
then we would have that
$\hat y_j \approx A_{k+1}z + A_{k+1} \delta_j + b_{k+1}$.

Next, recall the purpose of the encoding $z$ is to be completely decorrelated
from the input $x$.
So let us suppose this is the case, and that
$z \sim N(0, \sigma)$. Because the matrix $A_{k+1}$ was also
sampled i.i.d. from a Gaussian, we should have that
$\hat y_j = \delta'_j + c$.
where $c$ is a roughly-Gaussian vector.
Therefore, by taking the inner product of two encoded patches,
if the patches were placed in the same location their mean should
be dominated by the position encoding.
If, conversely, they were not placed in the same position than
the inner product should be approximately zero.

\subsubsection{Recover global permutation}
\label{s5}

We now show how to recover the final global permutation that maps
all locally-unpermuted encoded images onto the original positions.
Because all encoded images have now been aligned, all we must
do here is search for a single permutation $\rho$ by solving \\
$\text{arg max}_\rho \sum_i \sum_j \psim\big(x_i^{(j)}, y_{m(i)}^{(\rho(i))}\big).$
Again, we solve this by with max weight matching as we
have done all prior times.

Importantly, note that while this step will achieve good results
given the correct matching $m(\cdot)$, we have a (poor)
initial solution to the matching problem at this point.
As a result, there will be a significant amount of noise in the
computation, however his noise should be uncorrelated, and so the
signal from the correctly-matched images should dominate.

\subsubsection{Improve image-encoding matches}
\label{s6}

The final step in our attack generates an improved matching
between the original images and the encodings, making use of
the fact that (an approximation of) the complete image-patch-permutation $\pi_j$
has been recovered.
The core of the algorithm remains the same as the initial
matching: we compute the similarity between all images and
all encodings, and choose the matching that minimizes total
cost.

However, because we now know the correct permutation between
the images and the encodings, we can get a more accurate
measurement of the similarity between any given image
and encoding.
Instead of having to compute the min cost matching with
$\max_{\rho \in Sym(a^2)} \isim_{\rho}$, we can instead
just directly use the correct permutation $\rho^*$ that
we have recovered and compute the similarity as
$\isim_{\rho^*}$.

Note again that here we are assuming that we have the correct
alignment between each image and its corresponding encoding.
In general we will not have this perfect alignment, but it
will be a good approximation of the correct alignment.
And this results in a better matching from original images
to encodings.
Once we have this improved matching, we can then iterate these
last two steps progressively improving the generated matching
and generated alignment.
This perfectly solves the \nc challenge.

\subsection{Completely break \nc}
\subsubsection{Extract transformation network}
\label{s7}

Given this approximately correct matching $m$, and the approximately
correct image-patch-permutation $\pi_j$,
we will now solve for (an approximation of)
the original encoding function $T$.
%

Suppose that (1) our matching between original and encoded images
was perfect, and
(2) we have also recovered the image-level permutations perfectly.
(In practice it's not perfect, but for now just suppose it was.)

Then it would be possible to use this ``known plaintext'' to
extract the transformation $T$ by solving for it via gradient
descent.
Specifically, we can randomly initialize our own transformation
function $\tilde T$ from random initial weights, and then
via gradient descent solve
\[\mathop{\text{arg min}}_{\tilde T \in \text{Sim}(\mathcal{X})}\,\,\,\,\,
\sum_{i=1}^N  \lVert \tilde T(x_i) - y_{m(i)} \rVert.\]
%

This is possible to do because the system is over-determined:
a single image-encoding pair maps a
$256 \times 256 \times 3$ image to a $256 \times 2048$ dimensional
encoding, giving \emph{half a million} known input-output pairs.
Because the transformation network $T$ has roughly $30$ million parameters,
just $60$ image-encoding pairs are sufficient to completely determine
the transformation.

Notice, though, that we don't actually have a perfect matching $m$.
However, neural network training is exceptionally robust to label
noise.
As long as there exist a sufficient number of correctly aligned images,
the fact that many are not will not harm the quality of the solution
(by much).
This allows us to still run the model extraction on the noisy data.

\subsubsection{Recover perfect matching}
\label{s8}

Given the extracted transformation $\tilde T$, we can now generate our
\emph{expected} encodings $\tilde y_{m(i)} = \tilde T(x_{i})$.
This now gives us a significantly better way to match original
images to encoded images:
create a new similarity graph where the weight on each edge ($i,j$)
is given by $\lVert \tilde y_i - y_j \rVert$.
Solving the bipartite matching here again finds a solution that is
perfect for all experiments we have attempted.
(In fact, we can solve the \nc challenge perfectly without even resorting
to these last two steps.)

\section{Theoretical Insights}
The previous work of~\cite{carlini2020attack} showed the existence of  barriers against achieving private machine learning through \emph{instance encoding} schemes. However, the formulation of instance encoding by~\cite{carlini2020attack}  does not allow the encoding scheme to use private keys.
Here, we  focus on the setting where instance encoding mechanisms are allowed to use keys. We first introduce notions about such instance encoding mechanisms.
\paragraph{Notation.} We use calligraphic letter (e.g. $\cD$) for distributions and capital letters for sets.  We use $\cD\equiv \cD'$ to denote that two distributions $\cD$ and $\cD'$ are identical. We also use $\cD^{\times n}$ to denote the distribution of vectors of size $n$ with elements identically and independently distributed equivalent to $\cD$.  For a function $f$ and a distribution $\cD$, we use $f(\cD)$ to denote the imposed distribution of sampling from $\cD$ and then applying $f$. For a distribution $\cD$ and a function $c$, we use $\cD_c$ to denote  $f_c(\cD)$ where $f_c(x)=(x,c(x))$. For a function $h$, we use $\Risk(h,\cD_c)$ to denote $\Pr_{(x,y)\gets \cD_c}[h(x)\neq y]$. We say a learning algorithm has $(\eps,\delta)$ error  on a concept function $c$ and a distribution $\cD$, if for all $n\in \N$ we have $\Pr_{S\gets \cD_c^{\times n}}[\Risk(h,\cD_c)\geq \eps(n)]\leq \delta(n)$.  

For a specific value $y$ in the support of $c(\cD)$, we use $\cD_c^y$ to denote $f_c(\cD^y)$ where $\cD^y$ is the distribution $x \gets \cD$ conditioned on $c(x)=y$. 
We say a learning algorithm has $(\eps,\delta)$ \emph{balanced-error} on a concept function $c\colon X\to Y$ and a distribution $\cD$, if for all $n\in \N$ and all labels $y\in Y$ we have $\Pr_{S\gets \cD_c^{\times n}}[\Risk(h,\cD_c^y) )\geq \eps(n)]\leq \delta(n)$.
\begin{definition}[Learning with keyed encoding] \label{def:enc}
A \emph{learning  with keyed encoding} protocol consists of a pair of algorithms $L$ and $E$ as follows. The encoding mechanism $E$ is a potentially randomized algorithm $E\colon X\times K \times \AUX \to \widetilde{X}$  that takes an instance $x$, a key $k\in K$, and an auxiliary information $\aux\in \AUX$ as input and then outputs an encoded instance $ \widetilde{x}  \in \widetilde{X} $.   (When clear from the context, we omit the auxiliary information from the encoding input.) The learning algorithm $L\colon (\widetilde{X}\times Y)^*\to \Theta$   takes an encoded dataset and outputs a model $\theta\in\Theta$. 
We define the following properties for such a protocol $(L,E)$. 
    \begin{itemize}

      \item \textbf{Statistical NIA  privacy for a given concept class:} The encoding algorithm $E$ is $\eps$-NIA (no instance attack) private on a concept class $C$ and distribution $\cD$ if for all $c\in C$ the advantage of any adversary in the following game is bounded by $\eps$. The adversary $A$ selects two instances $x_0$ and $x_1$ such that $c(x_0)=c(x_1)$. Then the encoder samples a bit $b\gets \set{0,1}$ and a key $k\in K$ is sampled according to distribution $\mathcal{K}$ (which, without loss of generality, can be assumed to be uniform) and runs the encoding  $ E(x_b,k,c)$  to get $\widetilde{x}$. The adversary will be given $\widetilde{x}$, and it must decide whether $b=0$ or $b=1$ by outputting $b'$. The advantage of the adversary (for $c$ and $\cD$) is defined as $\Pr[b=b']-1/2$.   
              \item   \textbf{Statistical CIA privacy for a given a concept class:} 
    The $\eps$-CIA (chosen instance attack) privacy is defined similarly to statistical NIA   privacy with the difference that 
    after proposing the challenge instances $x_0,x_1$, the adversary gets access to oracle    $E'(\cdot, k,c)$, where $E'(\cdot, k,c)$ is the same as $E(\cdot, k,c)$ with the exception that none of  $x_0,x_1$ can be requested. This definition is different in that the adversary can query the encoding of any given point and is not restricted to random access to the encoded distribution.  
    
    \end{itemize}


\end{definition}

 Definition \ref{def:enc}  is aligned with Challenge 1 of \nc in the sense that it allows the adversary know the challenge instance $x_0,x_1$ and wants to match the encoded instance to $x_b$. In Challenge 1, the task is even harder, as \emph{multiple} instance (and not even selected by the adversary) shall be \emph{all} matched to their corresponding encodings, and even \emph{without} the oracle access provided to CIA attacks. We also comment that Definition \ref{def:enc} is   closely follows the   style of standard indistinguishability-based security definitions for encryption \cite{goldwasser1984probabilistic}, in which the job of the adversary is to map the given challenge ciphertext to the right plaintext that is known to the adversary.
Finally, we comment that one can also define a notion of \emph{random} instance attacks, that falls between NIA and CIA attacks, in which the adversary can request encodings of \emph{random instances}.

We  make the observation that  in the setting of CIA privacy, the distinguishing attacks presented in \cite{carlini2020attack} against unkeyed instance encoding schemes also apply to keyed instance encoding mechanisms. This can be verified in a straightforward manner, and hence we skip repeating such results here.

\paragraph{The ideal   scheme of \cite{yala2021neuracrypt}.} The authors of \cite{yala2021neuracrypt}   introduce an ``ideal encoding'' that is the basis of their ideal privacy goal.  This scheme first randomly samples a permutation $k\colon X\to X$ on the input space that maps each instance $x$ to an instance $x'$ with the same label $c(x)$.  
in which each $x$ is mapped to a random $\widetilde{x}$ with the same label $c(x)$ through a random permutation.

Here we    observe that the ideal scheme of \nc satisfies a very strong security guarantee, as it is  $0$-CIA private.  The reason  is that if the adversary proposes $x_0\neq x_1$, then for any pair of  encodings $\widetilde{x}_0 \neq \widetilde{x}_1$ it is equally likely that $E(x_0,k,c)=\widetilde{x}_0,E(x_1,k,c)=\widetilde{x}_1$ or that $E(x_0,k,c)=\widetilde{x}_1,E(x_1,k,c)=\widetilde{x}_0$, and this equality holds \emph{even conditioned} on  any way of fixing   the encodings of all the points other than $x_0,x_1$. Since the oracle in the CIA security game does not answer encodings of $x_0,x_1$, hence the adversary will have exactly chance $1/2$ of winning the game, even if it is given all the encodings other than those of $x_0,x_1$.

\paragraph{The ideal  vs. the heuristic schemes of \cite{yala2021neuracrypt}.} 
The work \cite{yala2021neuracrypt}  also claims that their \nc scheme  can be seen as a heuristic instantiation of the above-mentioned ideal encoding.  However, as described above,  the ideal encoding mechanism of  \cite{yala2021neuracrypt} needs to take some information about the concept function $c$ as auxiliary information; otherwise, it does not have any information about the labels. Also, note that the ideal encoding can potentially provide accuracy on both the original and the encoded instances (as they are from the same space).
%
However, there are two differences between the ideal encoding and the heuristic algorithm. Firstly, the heuristic encoding instantiation of \nc does not depend on instances' labels and does not take any auxiliary information about the concept function (that is later tried  to be learned). Secondly, in contrast to the ideal encoding scheme, the \nc approach only provides accuracy on encoded data. These disparities could be seen as indications that the theoretical analysis for the ideal scheme might not carry over to the heuristic scheme of \nc. This leads to the question of {\it{whether any encoding mechanism can successfully instantiate the ideal encoding algorithm specified above.}}

In this work, we prove  that any instantiation of the ideal instance encoding mechanism requires ``too much'' auxiliary information about the concept function. More formally, we show how to ``extract knowledge'' of $c$ from the such encoders.\footnote{This is reminiscent of knowledge extraction in cryptography \cite{goldwasser1989knowledge,bellare92dening}.} Note that we could always give the description of the concept function $c$ as auxiliary information to the encoding algorithm, and then the encoder can use that to produce the ideal encoding.\footnote{This is true, if the encoding itself is seen as an information theoretic process, ignoring the computational aspects of computing this mapping.} However, in that case the learning process becomes obsolete, as the encoder starts off, while it knows the concept function already. Roughly speaking, we show that ``knowing $c$ prior to encoding'' is necessary to achieve what an ideal encoding does achieve. In particular, we show that if the encoding scheme satisfies two properties (satisfied by the ideal encoding of~\cite{yala2021neuracrypt}), then it is possible to extract $c$ from it.

\begin{definition}\label{def:ideal}
We call an instance encoding mechanism $E\colon X\times K\times C \to X$ an weakly-ideal encoder for concept class $C$ if for all $c\in C$ we have
\begin{enumerate}[leftmargin=*]
    \item $E(x,\cK,c)\equiv E(x',\cK,c)$ for all $x,x'\in X : c(x)=c(x')$, 
    \item $c(E(x,k, c))=c(x)$ for any $x\in X$, $c\in C$ and $k\in K$.
\end{enumerate}
\end{definition}
 The first condition above is equivalent to to satisfying $0$-NIA property.
 Moreover, the ideal encoding of \cite{yala2021neuracrypt} satisfies both properties of Definition \ref{def:ideal}, and that is why we refer to such schemes as weakly ideal.


\begin{theorem}\label{thm:single_concept} Consider a distribution $\cD$   on $X$, a concept class $C \subseteq \set{0,1}^X$ and a weakly-ideal private encoding $E$ for $C$. 

If a learning algorithm $L$ with $m$ samples obtains $(\sfrac{1}{2} - \eps,\delta)$ balanced-error for all concept $c\in C$,  over distribution $\widetilde{\cD}\equiv E(\cD,\cK,c)$ for constants $0<\eps,\delta<1/2$,  
then for all $\tau\in[0,1]$ and any given pair $x_0,x_1$ where $c(x_0)=0,c(x_1)=1$, there is an oracle-aided PPT algorithm ($c$-extracting predictor) $\Extract^{L,E(\cdot,\cdot,c)}\colon X\to Y$ such that
$$\Risk(\Extract^{L,E(\cdot,\cdot,c)}(x),\cD' )\leq \tau$$
for arbitrary distribution $\cD'$. In other words, $\Extract^{L,E(\cdot,\cdot,c)}$ is predicting the output of $c$ with high accuracy with only \emph{two} labeled examples $(x_0,0),(x_1,1)$. Moreover, the running time of $\Extract$ is $\poly(m, 1/\epsilon, 1/\delta, 1/\tau).$
\end{theorem}

Since the ideal scheme of \cite{yala2021neuracrypt} is also  weakly ideal, Theorem \ref{thm:single_concept} suggests that finding an instantiation of the ideal instance encoding of \cite{yala2021neuracrypt} is at least as hard as improving the state-of-the-art accuracy for the learning problem to $100\%$, on all  distributions, and using only two labeled samples.

\section{Statement from Authors}

We shared a draft of this paper with the NeuraCrypt authors, and
asked if they would like to provide a response to be included 
into our paper.
We have reproduced their reply verbatim here:

\begin{displayquote}
The main NeuraCrypt challenge remains unsolved. NeuraCrypt is designed for the setting where a hospital wishes to release their data (e.g., X-rays) for public training while protecting their raw data. In this case, an attacker has access to the hospital's encoded dataset, which is shared publicly, and may utilize any other public X-ray datasets.   Critically, attackers do not have access to the plaintext versions of the entire encoded dataset. This scenario exactly corresponds to “Challenge 2: Identifying T from distributionally matched datasets”, and we emphasize that this challenge remains unsolved.

As a stepping stone towards this challenge, we also proposed a simplified synthetic setting, “Challenge 1: Reidentifying patients from matching datasets''. In challenge 1, the attacker already has access to the entire plaintext version of the encoded data, and can leverage this in any way as the basis for their attack. This paper solves challenge 1, and demonstrates that our current instantiation of NeuraCrypt is not secure in this setting owing to our use of ReLU activations. More importantly, this attack setting does not reflect a real-world scenario, as the attacker would have no incentive to attack the encoded data if they already had access to all the plaintext samples. While it is realistic to allow the attacker to obtain a few plaintext images, either by actively trying to participate in the dataset or by incentivizing the participants in the dataset to share their private images, such settings remain closer to Challenge 2. We continue to encourage the community to develop new attacks for Challenge 2.

In Appendix A, this paper claims that Challenge 2 is a “ciphertext-only” setting, and that encoding schemes that directly release the raw private data could be secure under Challenge 2. Both of these statements are incorrect. We note that Challenge 2 is not a “ciphertext-only” setting as the attacker also has knowledge of both the plaintext data distribution and the encoder distribution. We have previously shown that this rich distributional information is sufficient to break linear encoding schemes [YEO+21b].  More importantly, encoding schemes that release raw private data are not secure under Challenge 2, as an attacker can either “reidentify the original data or recover the private NeuraCrypt encoder”  [YEO+21a] to solve the challenge.  In contrast to the hypothetical scheme shown in this paper (Appendix A.1), these two scenarios (i.e recover raw data or private encoder) are equivalent for NeuraCrypt as the encoding function can be recovered with a plaintext attack if the images were recovered (Appendix C.2  [YEO+21b], [YEO+21a]).

Our theoretical results demonstrate the existence of an optimal family of instance encoding functions that obtain perfect privacy under our threat model (Theorem 2). We note that this is not a uniqueness result (i.e other optimal encodings may exist). Our existence result motivates the search for an encoding and motivates our approximation with neural networks. Throughout the paper, we emphasized that our theoretical claims did not extend to our specific network architecture, and thus evaluated the method with adversarial attacks and invited the community to propose new attacks through the NeuraCrypt challenge.  To demonstrate that neural network approaches are “fundamentally broken” (as claimed in this paper), the authors must prove an impossibility result for the use of neural networks as encoders, which is not done in this paper.

NeuraCrypt does not claim to solve the entire field of private learning, and many important questions remain. These questions range from characterizing possible attacks to developing new theory to bound the privacy of specific neural network architectures. In turn, these developments will lead to improved algorithms. We thank the authors for their participation in Challenge 1 and look forward to attacks on the designed use-case of NeuraCrypt (Challenge 2) and on future versions of the simplified Challenge 1.
\end{displayquote}

\subsection{Our Reply}
The NeuraCrypt research paper states an ideal security definition. This security definition, if satisfied, implies Challenge 1 is secure: an adversary should not be allowed to match original to encoded images (on nontrivial data). NeuraCrypt does not satisfy this security definition as proposed in their paper. Standard and most practical security definitions in cryptography use an indistinguishability framework in which an adversary aims to match a challenge ciphertext to one of the \emph{known and chosen} plaintexts, while the adversary has access to an encryption (CPA) or perhaps even a decryption oracle (CCA). By breaking Challenge 1, we do something much stronger: perfectly matching a \emph{set} (rather than a pair) of known instances (\emph{not} chosen by the adversary) to their encodings, \emph{without} any oracle call to an encoding oracle.

The NeuraCrypt challenge on GitHub also states a second security challenge, which requires an adversary to find a good approximation of the original encoding function given only encoded outputs. We still believe this second challenge is not as meaningful. This is for two reasons.

First, we prove there exist completely insecure schemes that appear secure under Challenge 2. So even if a scheme was resilient to Challenge 2, this would not imply it has any meaningful security. The authors call our simple provable claim incorrect, without saying where the gap is.

Second, consider the following hypothetical scenario in symmetric key cryptography research. An encryption algorithm E is proposed, along with an ideal security definition that is comparably stronger than being IND-CPA secure. An attack breaks the IND-CPA security (without even using the encryption oracle), but is rejected on grounds that the scheme is not yet broken in a ciphertext-only setting[footnote: Note ciphertext-only security here means the adversary does not have aligned plaintext-ciphertext pairs. An attack is still ciphertext-only if the adversary knows some auxiliary additional information, such as the plaintext data distribution.] where the adversary has to recover the key. This argument would not be reasonable, so we do not investigate Challenge 2 further in our paper.

Most importantly, the above response states that "the authors must prove an impossibility result”. We have done so: Theorem 3.3 proves that any approximation of the “ideal encoding” (which is the only scheme satisfying their security definition in the NeuraCrypt paper) requires the knowledge of the (concept) function that is being learned. This impossibility result shows that if one has access to this encoding mechanism, they can efficiently extract the concept function from it, essentially without any data. The response does not address this fact.

\section{Conclusion}

An instance encoding privacy scheme has clear and significant privacy
benefits.
It would, for example, allow users to share private datasets to cooperatively
train models without relying on a trusted third party.
Unfortunately, designing a secure instance encoding scheme has many
potential pitfalls, and we find that the \nc scheme does not satisfy its
privacy claims.
We believe that there are two important lessons from our attack:

\paragraph{Privacy schemes must come with privacy definitions.}
In order to be evaluated, methods that claim to give some amount of privacy
must precisely state what privacy is being offered.
While the \nc paper gives a definition of what it means for a scheme to be
``perfectly private'' in an idealized version of the setup, it has no definition of what
privacy the concrete \nc instantiation is designed to offer.
We therefore analyze \nc under the closest definition to its ideal privacy statements found in the \nc Challenge.

\paragraph{Neural networks aren't ideal functionalities.}
Our attack would be infeasible on functions that
truly behaved randomly. But in order for there to be any utility,
the \nc encoding operation must not destroy information completely.
Indeed, the fact that it is possible to train a neural network on
these encodings hints at the fact that neural networks should be able
to distinguish between encoded images in the first place.
We make use of this in our practical attack on \nc.

\paragraph{Iterative boosting can give complete breaks.}
Our attack works by taking an attack that detects if patches are
related with $70\%$ probability, and boosts this into an attack that
detects if images are similar with $98\%$ probability---then this gives
a small partial break on the scheme recovering $2\%$ of the matching, 
which we then boost into a fairly strong break recovering $50\%$ of
the matching, which we yet again boost into a complete break recovering
$100\%$ of the matching.

The fact that any individual step in our attack is possible is not surprising.
Indeed, when developing our attacks, we developed them in sequence from
top down.
Cryptographic systems are often discarded immediately upon the demonstration
of \emph{any} weakness, however limited, because of the understanding
that attacks only improve over time.
What starts out as a small weakness often finds a way to expand into
a complete break.

In the future, we hope that demonstrating even slight weaknesses will
be enough to cause researchers to abandon potential defenses.
There is often a not insignificant amount of effort that goes into turning
a small break into a complete one.
And while it is helpful to show that this can be done, we argue schemes
should be considered broken at the indication of the first weakness
as is done in cryptography.

\paragraph{Lessons.}
This is the second time we have developed a complete break on
an instance-encoding scheme.
Defenses that intend to deliver privacy through instance encoding \emph{must}
contain careful theoretical arguments---and not just about
ideal versions of their schemes (as was done in \nc)
or particular sub-problems used in their scheme (as was done in InstaHide).

There is no room for error in privacy.
Unlike in security, where zero-day vulnerabilities can be patched to mitigate harm, once a dataset has been published it must
remain private \emph{essentially forever}: an attack on the
scheme, even a decade later, can cause significant harm.
Our attack, for example, can only break the system given a (small)
number of known-plaintexts; we can not perform a ``ciphertext-only'' attack.
And so \emph{for now}, \nc can safely be used in settings where only
encoded images are released.
However if a hospital were to release just the encoded images using
\nc today, a future attack that extended ours to this new ciphertext-only setting
would violate the privacy of these encoded images retroactively.

Especially for schemes explicitly designed to protect medical images,
we fundamentally disagree with the research direction that aims
for best-effort privacy without strong proofs and rigorous evaluations.
If and when future schemes are proposed, it should be assumed
that they are just as fundamentally broken as InstaHide and
\nc, unless the authors are able to give strong and compelling
evidence to the contrary.

\bibliographystyle{alpha} 
\bibliography{neuracrypt}

\newcommand{\etalchar}[1]{$^{#1}$}
\begin{thebibliography}{YEO{\etalchar{+}}21b}

\bibitem[BG92]{bellare92dening}
Mihir Bellare and Oded Goldreich.
\newblock On dening proofs of knowledge.
\newblock In {\em Proceedings of CRYPTO}, volume~92, 1992.

\bibitem[BK98]{biryukov1998differential}
Alex Biryukov and Eyal Kushilevitz.
\newblock From differential cryptanalysis to ciphertext-only attacks.
\newblock In {\em Annual International Cryptology Conference}, pages 72--88.
  Springer, 1998.

\bibitem[CDG{\etalchar{+}}21]{carlini2020attack}
Nicholas Carlini, Samuel Deng, Sanjam Garg, Somesh Jha, Saeed Mahloujifar,
  Mohammad Mahmoody, Shuang Song, Abhradeep Thakurta, and Florian Tramer.
\newblock Is private learning possible with instance encoding?, 2021.

\bibitem[CLB{\etalchar{+}}19]{chen2019gmail}
Mia~Xu Chen, Benjamin~N Lee, Gagan Bansal, Yuan Cao, Shuyuan Zhang, Justin Lu,
  Jackie Tsay, Yinan Wang, Andrew~M Dai, Zhifeng Chen, et~al.
\newblock Gmail smart compose: Real-time assisted writing.
\newblock In {\em Proceedings of the 25th ACM SIGKDD International Conference
  on Knowledge Discovery \& Data Mining}, pages 2287--2295, 2019.

\bibitem[DDS{\etalchar{+}}09]{deng2009imagenet}
Jia Deng, Wei Dong, Richard Socher, Li-Jia Li, Kai Li, and Li~Fei-Fei.
\newblock Imagenet: A large-scale hierarchical image database.
\newblock In {\em 2009 IEEE conference on computer vision and pattern
  recognition}, pages 248--255. Ieee, 2009.

\bibitem[EKN{\etalchar{+}}17]{esteva2017dermatologist}
Andre Esteva, Brett Kuprel, Roberto~A Novoa, Justin Ko, Susan~M Swetter,
  Helen~M Blau, and Sebastian Thrun.
\newblock Dermatologist-level classification of skin cancer with deep neural
  networks.
\newblock {\em nature}, 542(7639):115--118, 2017.

\bibitem[GM84]{goldwasser1984probabilistic}
Shafi Goldwasser and Silvio Micali.
\newblock Probabilistic encryption.
\newblock {\em Journal of computer and system sciences}, 28(2):270--299, 1984.

\bibitem[GMR89]{goldwasser1989knowledge}
Shafi Goldwasser, Silvio Micali, and Charles Rackoff.
\newblock The knowledge complexity of interactive proof systems.
\newblock {\em SIAM Journal on computing}, 18(1):186--208, 1989.

\bibitem[HPQ{\etalchar{+}}18]{hosny2018artificial}
Ahmed Hosny, Chintan Parmar, John Quackenbush, Lawrence~H Schwartz, and
  Hugo~JWL Aerts.
\newblock Artificial intelligence in radiology.
\newblock {\em Nature Reviews Cancer}, 18(8):500--510, 2018.

\bibitem[HSLA20]{huang2020instahide}
Yangsibo Huang, Zhao Song, Kai Li, and Sanjeev Arora.
\newblock Instahide: Instance-hiding schemes for private distributed learning.
\newblock {\em ICML}, 2020.

\bibitem[IRK{\etalchar{+}}19]{irvin2019chexpert}
Jeremy Irvin, Pranav Rajpurkar, Michael Ko, Yifan Yu, Silviana Ciurea-Ilcus,
  Chris Chute, Henrik Marklund, Behzad Haghgoo, Robyn Ball, Katie Shpanskaya,
  et~al.
\newblock Chexpert: A large chest radiograph dataset with uncertainty labels
  and expert comparison.
\newblock In {\em Proceedings of the AAAI conference on artificial
  intelligence}, volume~33, pages 590--597, 2019.

\bibitem[WYB{\etalchar{+}}10]{wernick2010machine}
Miles~N Wernick, Yongyi Yang, Jovan~G Brankov, Grigori Yourganov, and Stephen~C
  Strother.
\newblock Machine learning in medical imaging.
\newblock {\em IEEE signal processing magazine}, 27(4):25--38, 2010.

\bibitem[YEO{\etalchar{+}}21a]{nc_challenge}
Adam Yala, Homa Esfahanizadeh, Rafael G. L.~D' Oliveira, Ken~R. Duffy, Manya
  Ghobadi, Tommi~S. Jaakkola, Vinod Vaikuntanathan, Regina Barzilay, and Muriel
  Medard.
\newblock Neuracrypt challenge.
\newblock \url{https://github.com/yala/NeuraCrypt-Challenge}, 2021.

\bibitem[YEO{\etalchar{+}}21b]{yala2021neuracrypt}
Adam Yala, Homa Esfahanizadeh, Rafael GL~D' Oliveira, Ken~R Duffy, Manya
  Ghobadi, Tommi~S Jaakkola, Vinod Vaikuntanathan, Regina Barzilay, and Muriel
  Medard.
\newblock Neuracrypt: Hiding private health data via random neural networks for
  public training.
\newblock {\em arXiv preprint arXiv:2106.02484}, 2021.

\end{thebibliography}

\appendix
\section{Alternate Privacy Game}

\nc also offers a second ``real-world'' privacy definition that we believe is not meaningful.
We state it here for completeness.

\smallskip \noindent
\textbf{Setup.} 
As before, let $\mathcal{X}$ be a dataset and $\text{sym}(\mathcal{X})$ a family of encoding functions.
Only Alice knows $\mathcal{X}$; both Alice and Bob know $\text{sym}(\mathcal{X})$.

\smallskip \noindent
\textbf{Alice} proceeds as before sampling $\vec{x} \subset \mathcal{X}$ choosing a $T$ and encoding $y_i = T(x_i)$ sending
the resulting $\vec y$ to Bob. Alice also chooses $\vec{z} \subset \mathcal{X}$ as a held-out test set for later.

\smallskip \noindent
\textbf{Bob} studies the encodings $\vec y$ (and possibly auxiliary data $\mathcal{Z}$ not
overlapping with $\mathcal{X}$).
He produces a guess at a transformation function $T'$.

\smallskip \noindent
\textbf{Evaluation.} Bob's ``score'' is computed through the following procedure.
Alice refers back to the samples $\vec{z}$ chosen earlier,
and then computes both $y_i = T(z_i)$ and  $y'_i = T'(z_i)$.
Count the number of instances $y_i$ where $i = \mathop{\text{arg min}}_j \mathcal{D}(y'_j, y_i)$
where $\mathcal{D}$ is a distance metric.
%
%
%
%
If Bob recovers $T' \equiv T$ then $\mathcal{D}(y_i, y'_i) = 0$ and therefore will score perfectly.
If however $T'$ is completely unrelated to $T$ then on average Bob should score just $1$.

\subsection{Why the second challenge is not as meaningful}

\paragraph{Ciphertext-only security has few practical applications.}
In this second challenge, the
adversary is \emph{exclusively} given access to the encoded images $\vec y$ and
no access to any original images, and must use the encodings alone to recover the function $T$.
In the terminology of cryptography, this is exactly asking for a \emph{ciphertext-only
attack}. This setting that has been recognized in cryptography as completely unrealistic for several decades.

While ciphertext-only security might be considered ``realistic'', schemes in cryptography
are intentionally designed to be secure even in ``unrealistic'' situations.
The reason for this design decision is twofold.
First, attacks only improve: a chosen-plaintext break often can be converted
in a known-plaintext break which can then be converted into a ciphertext-only break \cite{biryukov1998differential}.
Since the purpose of a challenge is to understand the security properties of a
system, it is better to understand the worst-case behavior than be hopeful that there is a chance it may be secure in some weak setting.

Second, ciphertext-only robustness is, contrary to the challenge claims, not more useful in practice.
In any practical setting where \nc would be deployed, an adversary would trivially
be able to obtain at least \emph{known} (if not \emph{chosen}) training data by visiting a hospital,
receiving a medical scan of themself, and then viewing their own medical data.

\paragraph{Non-private schemes solve this challenge.}
Any scheme that claimed to be ``privacy-preserving'' should at the
very least satisfy the basic requirement that given $y$ it should not 
be possible to perfectly reconstruct $x$.
We show that there exist schemes that do not satisfy this basic reconstruction requirement,
and yet are ``private'' under the \nc ciphertext-only challenge definition.

Assume for simplicity that inputs are chosen from $X=\set{-1,1}^n$.
Let $k \colon X \to \set{-1,1}^{a^2}$ be a function chosen at random among all possible
functions from this domain to this range. Also, for simplicity, let the metric $\cD$ be the hamming distance.
Then define $E(x,k) = (x,k(x))$.

We now prove that no adversary can win the security game of Challenge 2 with probability better than $1/N + n\tilde{O}(1/a)$. By the linearity of expectation, the expected number of matches will be $1+ nN\tilde{O}(1/a)$ which can be arbitrary close to $1$ for sufficiently large $a$.

Now let $E': X\to X\times \set{-1,1}^{a^2}$ be adversary's guessed encoding. We decompose $E'$ into $(E'_1,k')$ where the range of $E'_1$ is $\set{-1,1}^n$  and range of $k'$ is $\set{-1,1}^{a^2}$. For an input  $z_i \not \in \vec x$, let $r_i =|E'_1(z_i) -z_i| $ we have (all the following probabilities are over the output of random function on new queries )
\begin{align*}
    p&=\Pr[\forall j\neq i: |E'(z_i) - E(k,z_i)|\leq |E'(z_i) - E(k,z_j)|] \\&= \sum_c \Pr[|E'(z_i) - E(k,z_i)|=c]\prod_{j\neq i}\Pr[|E'(z_i)-E(k,z_j)|\geq c]
    \\&\leq\sum_c \Pr[|k'(z_i)-k(z_j)|\geq c-n]^{N-1}\Pr[|E'(z_i) - E(k,z_i)|=c]
    \\&=\sum_c \Pr[|k'(z_i)-k(z_j)|\geq c-n]^{N-1}\Pr[|k'(z_i)-k(z_i)|=c-r_i] 
    \\&=\sum_c \frac{{a^2 \choose c-r_i}}{2^{a^2}}\left[\sum_{c'\geq c-n} \frac{{a^2 \choose c'}} {2^{a^2}}\right]^{N-1}=q
\end{align*}

Now we can use normal approximation of binomial distributions to estimate the above probability. Let $t=a^2/2 - a\ln(a/n) $, $h=c-n$ and $g=c-r_i$. Then we have

\begin{align*}
q&\approx \int_\R e^{-(2g - a^2)^2/2a^2}  (1-\Phi ((2h-a^2)/2a))^{N-1} dc\\
&\leq \int_{c<t} e^{-(2g - a^2)^2/2a^2}  (1-\Phi ((2h-a^2)/2a))^{N-1} dc\\ &~+\int_{c\geq t} e^{\frac{-(2g - a^2)^2}{2a^2}}  (1-\Phi ((2h-a^2)/2a))^{N-1} dc\\
&\leq \frac{n}{a}\\
&~+ \int_{c\geq t} e^{-(2g - a^2)^2  +\frac{(2h - a^2)^2}{2a^2}} e^{-\frac{2h - a^2)^2}{2a^2}}(1-\Phi (\frac{2h-a^2}{2a}))^{N-1} dc\\
&\leq\frac{n}{a}\\
&=\frac{n}{a}+  e^{-((2a\ln(a/n))^2  -(2a\ln(a/n) -n)^2)/2a^2} \frac{1}{N}\\
&\leq \frac{n}{a} + (1+O(\frac{n\ln(a/n)}{a})) \frac{1}{N} \leq \frac{n}{a} + O(\frac{n\ln(a/n)}{aN}) + \frac{1}{N}
\end{align*}

As a result, this challenge is not as meaningful in what it guarantees about privacy, because it does not even prevent the
possibility that a  scheme satisfying this security to reveal their training data completely.
While this particular challenge could be repaired so that
schemes that solved this challenge necessarily also prohibit reconstruction,
the first issue would still remain: ciphertext-only security is not meaningful.
\begin{remark}
Note that in the above analysis, the adversary is required to use a Boolean encoding function. We can relax this and allow the adversary to pick a real function by selecting the random function $k\colon X \to \set{-1,0,+1}^n$. 
\end{remark}
\begin{remark}
The above construction uses a  large key which is the description of a random function. It is possible to make this key small by using a psudo-random function. They key $k$ will be a key for the PRF and then the PRF will be used instead of the random function. This comes at the cost of being secure only against  computationally bounded adversaries. 
\end{remark}

\section{Omitted Proof and discussions}
We first prove Theorem \ref{thm:single_concept}
\begin{proof}[Proof of Theorem \ref{thm:single_concept}]
The extraction algorithm $\Extract$ works as follows. Let $p$ be the fraction of positive examples in $\cD$. The algorithm first samples $m'\gets Binom(m,1-p)$. Then, given an instance $x_0$ from class $0$ and an instance $x_1$ from class $1$, it samples $m$ different keys $k_1,\dots,k_m$ and obtains encodings $e_1=E(x_0,k_1),\dots,e_{m'}=E(x_0,k_{m'}),e_{m'+1}=E(x_1,k_{m'+1}),\dots,e_m=E(x_1,k_m).$ Then, it uses $L$ to train classifier on the labeled dataset $\set{(e_1,0),\dots,(e_{m'},0),(e_{m'+1},1),\dots,(e_m,1)}$. Based on the assumption on $L$ and the properties of the encoding algorithm, this classifier will obtain balanced accuracy at least $0.5+\eps$ with probability at least $1-\delta$ on distribution $\widetilde{D}$. The algorithm repeats this process until it finds a classifier $h$ with balanced accuracy at least $0.5+\eps$. Now, at inference time, for a given sample $x\gets \cD'$, the predictor $\Extract$ first samples a set of keys $k^I_1,...,K^I_T\gets K$ and encodes $x$ with all the keys to get $e^I_i=E(x,k^I_i)$. Then it feeds these encodings to the classifier $h$ and takes the majority vote. Note that because of the balance accuracy of $h$ and also the properties of the encoding, we know that prediction of each of these encodings would be correct with independent probability at least $0.5+\eps$.  This means with enough repetition we can ensure that the prediction accuracy of each example is more than $1-\tau$.
\end{proof}
\begin{remark}
We can define an approximate version of ideal encoding of Definition \ref{def:ideal}. Instead of exact equality in the first condition, we can require the distributions to have small statistical distance. Or they can be defined to be computationally indistinguishable. We can also define the second condition to happen with high probability. In all this approximations, we can have a similar theorem to Theorem \ref{thm:single_concept} with a small degradation in the accuracy of the final classifier. 
\end{remark}
\paragraph{Further limitations in the multi-party setting.} Finally, one might conjecture that \nc provides a form of ``multiparty delegation'' for private model training scheme as follows. (1) First, multiple parties can encode their data sets $X_1,X_2,\dots$ into $E(X_1,k_1),E(X_2,k_2)\dots$ using their private encoding keys $k_1,k_2,\dots$. (2) Then a central powerful party (e.g., a cloud service provider) trains a model $h$ on the encoded data. (3) Finally, each of the parties, \emph{knowing their secret encoding key}, can use the trained model $h$. We observe that any such (purported) scheme, at the very least, cannot provide a security level as provided by multi-party computation schemes. The reason is that the encoded data $E(X_1,k_1),E(X_2,k_2),\dots$ would provide  ``free accuracy boosting'' to parties without the private keys. In particular, suppose a scheme as described above exists. Then, consider an adversary $A$ who has its own data set $X$, which if used as training set would only provide low accuracy. Then, $A$ can simply ``encode'' its own data into $E(X,k)$ using its own key $k$, add this batch of encoded data to the shared public pool of encoded data to get $S = E(X,k) \cup E(X_1,k_1) \cup E(X_2,k_2)\dots$, and then use $S$ to train its model.\footnote{At a high level, this attack can be interpreted as the observation that multi-key homomorphic encryption schemes cannot be decryptable using \emph{individual} decryption keys (as opposed to requiring \emph{all} of them).}

\end{document}